\documentclass[11pt, a4paper, UKenglish]{article}

\usepackage{booktabs} 
\usepackage[noend,ruled]{algorithm2e} 

\SetAlFnt{\small}
\SetAlCapFnt{\small}
\SetAlCapNameFnt{\small}
\SetAlCapHSkip{0pt}
\IncMargin{-\parindent}

\usepackage{hyperref}
\usepackage[utf8]{inputenc}
\usepackage[margin=1in]{geometry}

\usepackage{dsfont}
\usepackage{amsmath,amsthm, amsfonts, amssymb}
\usepackage{mathtools}
\usepackage{thmtools, thm-restate}
\newtheorem{lemma}{Lemma}

\newtheorem{claim}{Claim}

\theoremstyle{definition}
\newtheorem{definition}{Definition}

\usepackage{natbib}
\usepackage{xcolor}

\newcommand{\OPT}{\mathrm{OPT}}
\newcommand{\ALG}{\mathrm{ALG}}
\newcommand{\REV}{\mathrm{rev}}

\newcommand{\Ex}[2][]{\mbox{\rm\bf E}_{#1}\left[#2\right]}
\renewcommand{\Pr}[2][]{\mbox{\rm\bf Pr}_{#1}\left[#2\right]}

\title{Simplified Prophet Inequalities for Combinatorial Auctions}

\author{Alexander Braun \footnote{ alexander.braun@uni-bonn.de \newline \hspace*{0.58cm}The author has been funded by the Deutsche Forschungsgemeinschaft (DFG, German Research Foundation), \hspace*{0.58cm}Project No. 437739576.} \qquad Thomas Kesselheim \footnote{ thomas.kesselheim@uni-bonn.de} \\{\small Institute of Computer Science, University of Bonn, Germany}}

\begin{document}
	\maketitle
	\begin{abstract}
		We consider prophet inequalities for XOS and MPH-$k$ combinatorial auctions and give a simplified proof for the existence of static and anonymous item prices which recover the state-of-the-art competitive ratios. 

Our proofs make use of a linear programming formulation which has a non-negative objective value if there are prices which admit a given competitive ratio $\alpha \geq 1$. Changing our perspective to dual space by an application of strong LP duality, we use an interpretation of the dual variables as probabilities to directly obtain our result. In contrast to previous work, our proofs do not require to argue about specific values of buyers for bundles, but only about the presence or absence of items.

As a side remark, for any $k \geq 2$, this simplification also leads to a tiny improvement in the best competitive ratio for MPH-$k$ combinatorial auctions from $4k-2$ to $2k + 2 \sqrt{k(k-1)} -1$. 
	\end{abstract}
	\vspace*{0.8cm}
	\section{Introduction}
\label{section:introduction}

Prophet inequalities are an important tool to understand posted-pricing mechanisms for combinatorial auctions. Originally introduced as stopping time theorems in the 1970s and 80s by \citet{krengel1977semiamarts, krengel1978semiamarts} and \citet{Samuel84}, they regained attention a decade ago by the work of \citet{DBLP:conf/aaai/HajiaghayiKS07} and \citet{DBLP:conf/stoc/ChawlaHMS10} due to their interpretation in the context of algorithmic mechanism design. In particular, the competitive ratio of a pricing-based prophet inequality directly implies a DSIC and IR posted-prices mechanism with the same approximation guarantee. In the past years, our understanding of pricing-based prophet inequalities rose constantly and significantly by a line of inspiring work, see e.g. \citet{DBLP:conf/soda/FeldmanGL15, DBLP:conf/focs/DuettingFKL17,DBLP:conf/ec/GravinW19, DBLP:conf/focs/DuttingKL20} and \citet{DBLP:conf/ipco/CorreaCFPW22} among others. 

For XOS and MPH-$k$ combinatorial auctions \citep{DBLP:conf/soda/FeldmanGL15, DBLP:conf/focs/DuettingFKL17}, the common approach so far was to state a vector of static and anonymous item prices. Agents arrive one-by-one and can choose their most desired bundle among the remaining items. In order to derive a desirable competitive ratio with respect to the expected offline optimum, the algorithm designer is required to define prices carefully and tailored with respect to the class of valuation functions under consideration. In \citet{DBLP:conf/soda/FeldmanGL15}, item prices for XOS combinatorial auctions were set to be half of the expected contribution of an item to the social optimum. Later, \citet{DBLP:conf/focs/DuettingFKL17} discovered that it is sufficient to argue in the full information setting in contrast to dealing with random valuation profiles. 
Still, also their arguments are involved and require a deep understanding of the underlying habits. In contrast, easily accessible ideas concerning pricing-based prophet inequalities in combinatorial auctions are rare in the literature. Therefore, we try to advance our understanding of the following question: \\

{\emph{\centering What is the simplest way to prove prophet inequalities for combinatorial auctions?}} \\

As our contribution, we derive simplified proofs for prophet inequalities in XOS and MPH-$k$ combinatorial auctions which work as follows: First, we make use of a linear program with variables $p_j$ for any item $j$ corresponding to the static and anonymous item prices. For a variable assignment $(p_j)_{j \in M}$, the objective function of the LP will be non-negative if the posted-pricing mechanism with price vector $(p_j)_{j \in M}$ achieves a competitive ratio of $\alpha \geq 1$ (Section~\ref{subsection:lp_formulation}). To show that there exists a solution to the LP whose corresponding prices achieve the competitive ratio, we want to argue that there always is a feasible primal solution with non-negative objective value. To prove this, we use strong LP duality and change our perspective into dual space. We interpret dual variables and constraints (Section~\ref{subsection:understanding_dual}) and show that any dual feasible solution always has a non-negative objective value (Section~\ref{section:proof_competitive_ratios}). In particular, interpreting dual variables as probabilities over subsets of items, the dual constraints give a bound on the probability with which an item can be absent. Having this interpretation, we can argue by a reformulation of the dual objective that it will be non-negative for any dual feasible solution. As a consequence, at least the optimal primal solution also has to have a non-negative objective value. This implies that the corresponding prices $(p_j)_{j \in M}$ lead to a posted-pricing mechanism with competitive ratio $\alpha$.

Comparing this to previous work, we can also interpret the static and anonymous item prices from \citet{DBLP:conf/soda/FeldmanGL15} and \citet{DBLP:conf/focs/DuettingFKL17} in our context. Their prices are actually solutions to the linear program which lead to a non-negative objective value. Still, as mentioned already, proofs for competitive ratios in the work of \citet{DBLP:conf/soda/FeldmanGL15} and \citet{DBLP:conf/focs/DuettingFKL17} are based on arguments about a specific choice of $(p_j)_{j \in M}$. In particular, they argue about prices and valuation functions and show that, on the one hand, prices are sufficiently high to cover the welfare loss induced by allocating an item. On the other hand, they need to argue that prices are low enough such that agents are willing to buy. In contrast, our approach can avoid any argument on specific buyers' valuations at all. 

Concerning competitive ratios, our LP based proof recovers the tight guarantees of $2$ for a single item as well as XOS valuation functions \citep{DBLP:conf/soda/FeldmanGL15} and the currently best known $4k-2$ for MPH-$k$ valuations \citep{DBLP:conf/focs/DuettingFKL17}. As a side remark, for any $k \geq 2$, we also get a tiny improvement in the competitive ratio for MPH-$k$ prophet inequalities from $4k-2$ to $2k + 2 \sqrt{k(k-1)} -1$ which is slightly smaller than $4k-2$. 

	\section{Notation and Preliminaries}
\label{Section:Preliminaries}

We consider the following setting: There is a set of $m$ heterogeneous items $M$ and a sequence of $n$ agents arriving one-by-one. As agent $i$ arrives, we get to know her valuation function $v_i : 2^M \rightarrow \mathbb{R}_{\geq 0}$ and agent $i$ buys the bundle of (currently unassigned) items which maximizes her quasi-linear utility. Valuation functions are always non-negative and bounded for any bundle as well as monotone and normalized, i.e. $v_i(S) \leq v_i(S')$ for $S \subseteq S' \subseteq M$ and $v_i(\emptyset) = 0$. We assume that each $v_i \sim \mathcal{D}_i$ is drawn independently from a publicly known, not necessarily identical distribution $\mathcal{D}_i$ and denote by $\mathbf{v} \sim \times_{i=1}^{n} \mathcal{D}_i$ the valuation profile of all agents.

An \textit{allocation} $\mathbf{X} = (X_i)_{i \in [n]}$ is a vector of item bundles such that agent $i$ is allocated bundle $X_i$ and for two agents $i \neq i'$, we have $X_i \cap X_{i'} = \emptyset$. The \textit{social welfare} of an allocation $\mathbf{X}$ given valuation profile $\mathbf{v}$ is defined as $\mathbf{v}(\mathbf{X}) \coloneqq \sum_{i \in [n]} v_i(X_i)$. An algorithm outputting allocation $\mathbf{X}$ is \emph{$\alpha$-competitive} for some $\alpha \geq 1$ with respect to the expected offline optimal social welfare if $\Ex[\mathbf{v}]{\mathbf{v}(\mathbf{X}) } \geq \frac{1}{\alpha} \cdot \Ex[\mathbf{v}]{\max_{\mathbf{X}^\ast} \mathbf{v}(\mathbf{X}^\ast)}$.

\paragraph{Valuation Functions.} We consider XOS valuation functions which are defined as follows: A set function $a : 2^M \rightarrow \mathbb{R}_{\geq 0}$ is \textit{additive} if and only if there are numbers $c_1,\dots,c_{m} \in \mathbb{R}_{\geq 0}$ such that for any $S \subseteq M$ we have $a(S) = \sum_{j \in S} c_j$. A set function $v: 2^M \rightarrow \mathbb{R}_{\geq 0}$ is \textit{fractionally subadditive} (also called \textit{XOS}) if and only if there are additive set functions $a_1,\dots,a_t$ such that for every $S \subseteq M$ we have $v(S) = \max_{i \leq t } a_i(S)$. 

A natural extension of XOS to functions which admit complementaries are MPH-$k$ valuations \citep{DBLP:conf/aaai/FeigeFIILS15}. Consider a valuation function $v:2^M \rightarrow \mathbb{R}_{\geq 0}$. A hypergraph representation of the function $v$ is a set function which satisfies $v(S) = \sum_{X \subseteq S} w(X)$. We call a set $X$ with $w(X) \neq 0$ a hyperedge of $w$, and a positive-hyperedge if $w(X) > 0$. The \emph{rank} of a hypergraph representation $w$ is the cardinality of the largest hyperedge in $w$. If the hypergraph representation of $v$ only contains non-negative hyperedges, we call this a positive-hyperedge-$k$ function (PH-$k$) as introduced by \citet{10.1145/2229012.2229016}. The definition of the MPH-$k$ hierarchy now represents a valuation function $v$ as the maximum over a set of PH-$k$ functions.

\begin{definition} (Maximum-over-Positive-Hypergraph-$k$ class \citep{DBLP:conf/aaai/FeigeFIILS15}) \\
	A monotone valuation function $v : 2^M \rightarrow \mathbb{R}_{\geq 0}$ is MPH-$k$ if there is a set $\{v_\ell\}_{\ell \in \mathcal{L}}$ of PH-$k$ functions such that 
	$ v(S) = \max_{\ell \in \mathcal{L}} v_\ell(S)$ for any $ S \subseteq M$ and arbitrary index set $\mathcal{L}$.
\end{definition}
	\section{General Framework}
\label{section:framework_lp}

In order to derive proofs for our competitive ratios, we start with a lower bound on the social welfare of our algorithm. By $\OPT_i(\mathbf{v})$ we denote the (possibly empty) bundle of items which agent $i$ gets in the optimal allocation on valuation profile $\mathbf{v}$.

\begin{lemma}\label{lemma:framework_lp}
	For any combinatorial auction with monotone valuation functions, the social welfare of the sequential posted-prices mechanism with price vector $\mathbf{p} = (p_j)_{j \in M}$ fulfils for any $\beta \in [0,1]$ \[ \Ex[\mathbf{v}]{\mathbf{v}\left( \ALG(\mathbf{v}) \right)} \geq \min_{T \subseteq M} \left( \sum_{j \in T} p_j + \sum_{i=1}^n  \Ex[\mathbf{v}]{ \sum_{S \subseteq M} \beta \left( v_i(S\setminus T) - \sum_{j \in S \setminus T} p_j \right) \mathds{1}_{S = \OPT_i(\mathbf{v})} } \right) \enspace. \] 
\end{lemma}

The proof of this lemma follows standard steps in the prophet inequality literature \citep{DBLP:conf/soda/FeldmanGL15, DBLP:conf/focs/DuettingFKL17,DBLP:conf/ec/GravinW19, DBLP:conf/focs/DuttingKL20, DBLP:conf/ipco/CorreaCFPW22}. In particular, we split the social welfare into revenue and utility and bound each quantity separately. As a final step, we lower bound the social welfare by allowing the set of allocated items $T$ to be chosen by an adversary. For the sake of completeness, we provide a full proof of Lemma~\ref{lemma:framework_lp} in Appendix~\ref{section:appendix} which the experienced reader may skip. 

\subsection{An LP Formulation and its Dual}
\label{subsection:lp_formulation}

Having the lower bound on the social welfare obtained by the algorithm, we actually want to show that for any set $T \subseteq M$, the lower bound in Lemma~\ref{lemma:framework_lp} is at least as large as a $\frac{1}{\alpha}$-fraction of the expected offline optimum. Interpreting this as a constraint for any set $T \subseteq M$, we can formulate an LP which has a non-negative objective value whenever the desired competitive ratio $\alpha$ can be achieved. Thus, in order to prove a competitive ratio, we only need to argue about the LP. 

The LP has variables $p_j$ for any item $j \in M$ where $p_j$ corresponds to the static and anonymous item price for item $j$. In addition, there are slack variables $\ell_{+}$ and $\ell_{-}$ which indicate if the desired competitive ratio can be achieved or not.

\begin{align*}
	\text{max}\quad&\ell_{+}-\ell_{-}\\
	\text{s.t.}\quad& \sum_{j \in T} p_j + \sum_{i=1}^n \Ex[\mathbf{v}]{\sum_{S \subseteq M} \beta \left( v_i(S\setminus T) - \sum_{j \in S \setminus T} p_j \right) \mathds{1}_{S = \OPT_i(\mathbf{v})} } \\
	&\hspace*{50pt}\geq \frac{1}{\alpha} \Ex[\mathbf{v}]{\mathbf{v}(\OPT(\mathbf{v}))} +\ell_{+} - \ell_{-} &&\text{for all $T \subseteq M$}\\
	&p_j \geq 0 &&\text{for all $j \in M$}\\
	&\ell_{+}, \ell_- \geq 0 . 
\end{align*}

We note that \citet{DBLP:conf/focs/DuttingKL20} use a similar variant of this LP and its dual in order to show the existence of prices for subadditive combinatorial auctions\footnote{In contrast to our LP, they draw the bundle $S$ from some probability distribution whereas we set it equal to the bundle of items which agent $i$ gets in the offline optimum on $\mathbf{v}$. In addition, they do subtract the prices for all items in $S$ whereas in our formulation, it is essential to only consider prices of items in $S \setminus T$. Subtracting prices for any item in $S$ will result in a worse competitive ratio already in the case of a single item.}. Concerning the constraints for any set $T \subseteq M$, observe that they can be rearranged to 

\begin{align*}
	&\sum_{j \in M} p_j \left( \beta \sum_{i=1}^{n} \Pr[\mathbf{v}]{j \in \OPT_i(\mathbf{v})} \mathds{1}_{j \notin T} - \mathds{1}_{j \in T} \right) + \ell_+ - \ell_- \\ &\hspace*{50pt} 
	\leq \beta \sum_{i=1}^n \sum_{S \subseteq M} \Ex[\mathbf{v}]{  v_i(S\setminus T) \mathds{1}_{S = \OPT_i(\mathbf{v})} } -  \frac{1}{\alpha} \Ex[\mathbf{v}]{\mathbf{v}(\OPT(\mathbf{v}))} \enspace.
\end{align*}

In order to argue that the LP has a non-negative objective value, we can consider the dual program and use strong duality. In particular, we will argue that any feasible dual solution has an objective value which is non-negative. Via strong duality, this directly implies that at least the optimal primal solution has a non-negative objective value and hence, the corresponding prices lead to the desired competitive ratio. 

The dual of the LP introduced above has variables $\mu_T \geq 0$ for every set $T \subseteq M$ and is given by 

\begin{align*}
	\text{min}\quad&\sum_{T} \mu_T \left( \sum_{i=1}^{n} \beta \Ex[\mathbf{v}]{ \sum_{S \subseteq M} v_i(S \setminus T) \mathds{1}_{S = \OPT_i(\mathbf{v})} } - \frac{1}{\alpha}\Ex[\mathbf{v}]{\mathbf{v}(\OPT(\mathbf{v}))}\right)\\
	\text{s.t.}\quad&\sum_{T} \mu_T \left( \beta \sum_{i=1}^{n} \Pr[\mathbf{v}]{j \in \OPT_i(\mathbf{v})} \mathds{1}_{j \notin T} - \mathds{1}_{j \in T}  \right)   \geq 0 &&\text{for all $j \in M$}\\
	&\sum_{T} \mu_T = 1\\
	&\mu_T \geq 0 &&\text{for all $T \subseteq M$}.
\end{align*}

Having this, we are able to state the lemma which will simplify the proof of threshold-based prophet inequalities.

\begin{lemma}\label{lemma:lp_duality}
	For any combinatorial auction with monotone valuation functions, there exists a sequential posted-prices mechanism with price vector $\mathbf{p} = (p_j)_{j \in M}$ which is $\alpha$-competitive with respect to the expected offline optimum if the objective value of the dual program is non-negative for any feasible dual solution.
\end{lemma}

Observe that by the construction above, showing the existence of suitable prices boils down to arguing about the dual of a linear program. 

\subsection{Understanding the Dual Program} 
\label{subsection:understanding_dual}

Before we continue to derive our competitive ratios, we will start by gaining a better understanding of the dual. 

\paragraph{Dual Constraints.} First, note that $\sum_{T} \mu_T = 1$ and $\mu_T \geq 0$ for any $T$. Hence, we can also interpret the vector $(\mu_T)_T$ as a probability distribution over subsets $T\subseteq M$. 

Second, we can without loss of generality assume that $\sum_{i=1}^n \Pr[\mathbf{v}]{j \in \OPT_i(\mathbf{v})} = 1$ as the optimum will always allocate any item in any realization $\mathbf{v}$. As a consequence, we can reformulate the dual constraints \[ \sum_{T} \mu_T \left( \beta \sum_{i=1}^{n} \Pr[\mathbf{v}]{j \in \OPT_i(\mathbf{v})} \mathds{1}_{j \notin T} - \mathds{1}_{j \in T}  \right)   \geq 0 \] for any item $j \in M$ as follows: 

\begin{align*}
	& \sum_{T} \mu_T \left( \beta \sum_{i=1}^{n} \Pr[\mathbf{v}]{j \in \OPT_i(\mathbf{v})} \mathds{1}_{j \notin T} - \mathds{1}_{j \in T}  \right) \\
	& \hspace*{0.8cm} = \sum_{T} \mu_T \left( \beta  \mathds{1}_{j \notin T} - \mathds{1}_{j \in T}  \right) \\ 
	& \hspace*{0.8cm} = \beta \sum_{T} \mu_T \mathds{1}_{j \notin T} - \sum_T \mu_T \mathds{1}_{j \in T}  \\
	& \hspace*{0.8cm} = \beta \Pr[\mu]{j \notin T} - \Pr[\mu]{j \in T} \geq 0 \enspace,
\end{align*}
where we denote by $\Pr[\mu]{j \notin T}$ the probability that some item $j$ is not contained in a set $T \subseteq M$ sampled with respect to distribution $(\mu_T)_T$. In other words, $ \Pr[\mu]{j \notin T} \coloneqq \Pr[T \sim \mu]{j \notin T} \coloneqq  \sum_{T} \mu_T \mathds{1}_{j \notin T} $. 

Using that $\Pr[\mu]{j \notin T} + \Pr[\mu]{j \in T} = 1$, the dual constraints for any item $j \in M$ are equivalent to our \emph{first key property}:
\begin{align}\label{equation:dual_probabilites}
	\Pr[\mu]{j \notin T} \geq \frac{1}{1 + \beta} \hspace*{0.8cm } \text{ and } \hspace*{0.8cm } \Pr[\mu]{j \in T} \leq \frac{\beta}{1 + \beta} \enspace.
\end{align} 
We will later set $\beta =1 $ for a single item and XOS functions. In this case, Inequality~\eqref{equation:dual_probabilites} simply states that any item $j$ can only be in $T$ with probability at most $\frac{1}{2}$.

\paragraph{Dual Objective.} To make the dual objective more accessible, we change the order of summation:
\begin{align*}
\text{dual obj.} & = \sum_{T} \mu_T \sum_{i=1}^{n} \beta \Ex[\mathbf{v}]{ \sum_{S \subseteq M} v_i(S \setminus T) \mathds{1}_{S = \OPT_i(\mathbf{v})} } - \frac{1}{\alpha}\Ex[\mathbf{v}]{\mathbf{v}(\OPT(\mathbf{v}))} \\ 
& = \sum_{i=1}^{n} \sum_{S \subseteq M}  \Ex[\mathbf{v}]{ \mathds{1}_{S = \OPT_i(\mathbf{v})} \beta \sum_{T} \mu_T v_i(S \setminus T)  } - \frac{1}{\alpha}\Ex[\mathbf{v}]{  \sum_{i=1}^{n} \sum_{S \subseteq M} \mathds{1}_{S = \OPT_i(\mathbf{v})}  v_i(S )  } \\ 
& = \sum_{i=1}^{n} \sum_{S \subseteq M}  \Ex[\mathbf{v}]{ \mathds{1}_{S = \OPT_i(\mathbf{v})} \left( \beta \sum_{T} \mu_T v_i(S \setminus T) - \frac{1}{\alpha} v_i(S) \right) }
\end{align*}
Instead of arguing that the dual objective is non-negative, we will show the following for suitable choices of $\alpha$ and $\beta$: If a vector $(\mu_T)_T$ is feasible with respect to the dual, then the term 
\begin{align*} 
	\beta \sum_{T} \mu_T v_i(S \setminus T) - \frac{1}{\alpha} v_i(S)
\end{align*}
is non-negative for any $v_i$ and $S$. In particular, we will show the following equivalent claim for the respective choices of  $\alpha$ and $\beta$:
\begin{claim}\label{claim:inequality_dual_objective}
For any set $S$, any XOS/MPH-$k$ function $v_i$ and any dual feasible solution $(\mu_T)_T$, the following holds:
\begin{align} \label{equation:dual_objective_simplified}
\beta \sum_{T} \mu_T v_i(S \setminus T) \geq \frac{1}{\alpha} v_i(S) \hspace*{0.4cm } \text{ or equivalently for $\beta >0$ } \hspace*{0.4cm } \sum_{T} \mu_T v_i(S \setminus T) \geq \frac{1}{\alpha \beta} v_i(S) 
\end{align}
\end{claim}

This will be our \emph{second key ingredient}. \\
We can interpret it as follows: When drawing a set $T$ with respect to distribution $(\mu_T)_T$, the value that remains from some fixed set $S$ after removing $T$ is still at least a $\frac{1}{\alpha \beta}$-fraction of the original value $v_i(S)$.

Having our two key ingredients, we are ready to derive the competitive ratios. 
	\section{Deriving Competitive Ratios Easily}
\label{section:proof_competitive_ratios}

By the construction in Section~\ref{section:framework_lp}, we are only required to argue that the dual objective is non-negative for the choices of $\alpha$ and $\beta$ in the respective settings. As mentioned, we will show a stronger result, namely that Inequality~\eqref{equation:dual_objective_simplified} holds for any subset of items $S \subseteq M$ and any valuation function $v_i$ which is XOS or MPH-$k$. We summarize the respective statements first and give proofs for each afterwards.
	
\begin{lemma}\label{lemma:single_item}
	In the case of a single item, Inequality~\eqref{equation:dual_objective_simplified} in Claim~\ref{claim:inequality_dual_objective} holds for $\alpha =2$ and $\beta=1$.
\end{lemma}
		
\begin{lemma}\label{lemma:XOS}
	For XOS valuation functions, Inequality~\eqref{equation:dual_objective_simplified} in Claim~\ref{claim:inequality_dual_objective} holds for $\alpha =2$ and $\beta=1$.
\end{lemma}
			
Observe that these two lemmas directly correspond to the best known competitive ratios and are tight. In addition, note that the class of XOS valuation functions contains e.g. submodular functions and is equivalent to the class of MPH-$1$ valuations.
		
For MPH-$k$ valuation functions, we first give a simplified proof of the competitive ratio of \citet{DBLP:conf/focs/DuettingFKL17}. In particular, \citet{ DBLP:conf/focs/DuettingFKL17} introduced a reduction which allows to only argue about deterministic valuation functions instead of randomly drawn ones. Our proof will be even simpler: our argument only requires to consider the sizes of sets which are relevant in MPH-$k$ valuations, the valuations as such do not play a role at all.
			
\begin{lemma}\label{lemma:MPH_k_balanced_prices}
	For MPH-$k$ valuation functions, Inequality~\eqref{equation:dual_objective_simplified} in Claim~\ref{claim:inequality_dual_objective} holds for $\alpha =4k-2$ and $\beta=\frac{1}{2(k-1)}$ when $k \geq 2$.
\end{lemma}
				
Finally, for all $k \geq 2$, we show that we can get a tiny improvement in the competitive ratio compared to previously known results for MPH-$k$ valuations as $2k + 2 \sqrt{k(k-1)} -1 < 4k-2$.
				
\begin{lemma}\label{lemma:MPH_k_improved}
	For MPH-$k$ valuation functions, $\alpha =2k + 2 \sqrt{k(k-1)} -1$ and $\beta=\sqrt{\frac{k}{k-1}}-1$, Inequality~\eqref{equation:dual_objective_simplified} in Claim~\ref{claim:inequality_dual_objective} holds when $k \geq 2$.
\end{lemma}
				
\subsection{Warm-Up: A Single Item}	
\label{subsection:single_item_proof}			
					
Observe that in the case of a single item, the vector $(\mu_T)_T$ only has two entries $\mu_\emptyset$ and $\mu_{\{ \text{item} \}}$. 

\begin{proof}[Proof of Lemma~\ref{lemma:single_item}]

We overload notation and denote by $v_i$ the value of agent $i$ for the item in order to rewrite the left-hand side of Inequality~\eqref{equation:dual_objective_simplified} as follows:
\begin{align*}
	\sum_{T} \mu_T v_i(S \setminus T) & = \sum_{T} \mu_T \left( v_i  \mathds{1}_{S = \{ \text{item} \} } \mathds{1}_{T = \emptyset} \right) = \left( \sum_{T} \mu_T   \mathds{1}_{T = \emptyset} \right) v_i  \mathds{1}_{S = \{ \text{item} \} } \\ & =  \mu_\emptyset \cdot v_i(S)  =  \Pr[\mu]{T = \emptyset} \cdot v_i(S)
\end{align*}	
Hence, what remains is to show that $\Pr[\mu]{T = \emptyset} \geq \frac{1}{2}$ for any dual feasible solution. Indeed, observe that by Inequality~\eqref{equation:dual_probabilites} for $\beta =1$, we have \[ \Pr[\mu]{T = \emptyset} = \Pr[\mu]{\text{item} \notin T} \geq \frac{1}{2} \enspace. \qedhere \] 

\end{proof}	
					
Note that we did only argue about the value of the dual variables $\mu_T$ and did not at all need to take a specific value $v_i$ into account. In particular, it is only important that the probability distribution $\left( \mu_T \right)_T$ does not put too much mass on $T = \{\text{item}\}$ or in other words, $\mu_\emptyset$ is sufficiently large. The remarkable thing is that this argument nicely extends to XOS and MPH-$k$ valuation functions without becoming much more involved. In the more general settings, we will consider (subsets of) items separately and argue in an equivalent way about $\left( \mu_T \right)_T$.

\subsection{XOS valuation functions}	
\label{subsection:XOS_proof}

We proceed in a similar way as in the proof of Lemma~\ref{lemma:single_item}. Still, we need to take the combinatorial structure of the valuation functions into account. 

\begin{proof}[Proof of Lemma~\ref{lemma:XOS}]
	First, observe that for any XOS valuation function $v_i$, when fixing one of the additive supporting functions, we obtain a lower bound on the function's value. In particular, we can bound the value of $ v_i(S \setminus T) $ from below by considering the additive function with which buyer $i$ would evaluate the set $S$. We denote this function by $w_i^{S}$. Using this property, we can lower bound the left-hand side of Inequality~\eqref{equation:dual_objective_simplified} as follows:
	\begin{align*}
		\sum_{T} \mu_T v_i(S \setminus T) & \geq \sum_{T} \mu_T w_i^S(S \setminus T)  = \sum_{T} \mu_T \sum_{j \in S} w_i^S(\{j\}) \mathds{1}_{j \notin T} \\
		& = \sum_{j \in S} w_i^S(\{j\}) \sum_{T} \mu_T \mathds{1}_{j \notin T}  = \sum_{j \in S} w_i^S(\{j\})   \Pr[\mu]{j \notin T} \\
		& \geq  \frac{1}{2} \sum_{j \in S} w_i^S(\{j\})  
		 = \frac{1}{2} v_i(S)
	\end{align*}
	The first inequality holds due to the XOS property of $v_i$, the second step exploits that the function $w_i^S$ is additive. In the last inequality, we made use of Inequality~\eqref{equation:dual_probabilites} for $\beta =1$ to see that $\Pr[\mu]{j \notin T} \geq \frac{1}{2}$. \qedhere 

\end{proof}

\subsection{MPH-$k$ valuation functions}	
\label{subsection:mph_k_proof}

Also the proof for MPH-$k$ valuations follows a similar template. In contrast to XOS valuation functions, we have to take into account that items can complement each other. Further, we give a combined proof for Lemmas~\ref{lemma:MPH_k_balanced_prices} and \ref{lemma:MPH_k_improved} with general $\alpha \geq 1$ and $\beta \in (0,1]$ and use two observations afterwards to show the desired competitive ratios.

\begin{proof}[Proof of Lemmas~\ref{lemma:MPH_k_balanced_prices} and \ref{lemma:MPH_k_improved}]
	As in the case of XOS valuation functions, also for MPH-$k$ valuations we can fix one of the supporting PH-$k$ functions to obtain a lower bound on the value. In particular, we can bound the value of $v_i(S \setminus T)$ from below by only considering the PH-$k$ function with which buyer $i$ would evaluate the set $S$, denoted by $v_i^{S} (\cdot)$ with corresponding weights on hyperedges denoted by $w_i^{S} (\cdot)$. This implies the following lower bound on the left-hand side of Inequality~\eqref{equation:dual_objective_simplified}: 
	
	\begin{align*}
	 \sum_{T} \mu_T v_i(S \setminus T)  & \geq  \sum_{T} \mu_T v_i^S(S \setminus T)  \\ 
		& =  \sum_{T} \mu_T  \sum_{X \subseteq S} w_i^S(X) \mathds{1}_{X \cap T = \emptyset} \\
		& =  \sum_{X \subseteq S} w_i^S(X) \left(  \sum_{T} \mu_T  \mathds{1}_{X \cap T = \emptyset} \right) \\ 
		& = \sum_{X \subseteq S} w_i^S(X)   \Pr[\mu]{X \cap T = \emptyset}  
	\end{align*}
	Next, we argue that for the respective choices of $\alpha \geq 1$ and $\beta \in (0,1]$, the term $ \Pr[\mu]{X \cap T = \emptyset}$ is at least as large as $\frac{1}{\alpha \beta}$ for any $X$ with $|X| \leq k$ and any feasible dual solution. \\
	
	To this end, first note that by the union bound and Inequality~\eqref{equation:dual_probabilites} \[ \Pr[\mu]{\exists \ j \in X : j \in T } \leq \sum_{j \in X} \Pr[\mu]{j \in T} \leq \sum_{j \in X} \frac{\beta}{1+\beta} \leq \frac{k \beta}{1 + \beta} \enspace, \] where in the last inequality, we used that $|X| \leq k$ via the MPH-$k$ property. Hence, we can lower bound $\Pr[\mu]{X \cap T = \emptyset}$ via
	\begin{align*}
		 \Pr[\mu]{X \cap T = \emptyset}  = \Pr[\mu]{\forall \ j \in X : j \notin T }  = 1 - \Pr[\mu]{\exists \ j \in X : j \in T } \geq 1 - \frac{k \beta}{1 + \beta}  \enspace.
	\end{align*}
	As a consequence, we obtain that 
	\begin{align*}
		\sum_{T} \mu_T v_i(S \setminus T)  & \geq  \sum_{X \subseteq S} w_i^S(X)  \left(1 - \frac{k \beta}{1 + \beta} \right) = \left(1 - \frac{k \beta}{1 + \beta} \right) v_i(S) \enspace.
	\end{align*}
	We can conclude the proof of Lemma~\ref{lemma:MPH_k_balanced_prices} by observing that for $\alpha =4k-2$ and $\beta=\frac{1}{2(k-1)}$, we have $  1 - \frac{k \beta}{1 + \beta} = \frac{1}{\alpha \beta}$. \\ The proof of Lemma~\ref{lemma:MPH_k_improved} follows by checking that for $\alpha =2k + 2 \sqrt{k(k-1)} -1$ and $\beta=\sqrt{\frac{k}{k-1}}-1$, also $1 - \frac{k \beta}{1 + \beta} = \frac{1}{\alpha \beta} $ which again leads to the desired result.
\end{proof}	

One might wonder what is really going on. As a matter of fact, we made use of the property that the probability distribution $(\mu_T)_T$ cannot put too much mass on any item. In other words, the dual constraints ensure the following: When drawing a set $T$ with respect to $(\mu_T)_T$, any item $j \in M$ is ensured to be not in $T$ with a reasonably high probability. This implies for XOS and MPH-$k$ functions that the values of $v_i(S \setminus T)$ and $v_i(S)$ are sufficiently close to each other. 

	\section{Future Research}
\label{section:conclusion}

In contrast to previous work, our proof for the existence of thresholds can avoid arguments on specific buyers' valuations in comparison to the offline optimum. It would be interesting to see the power of LP duality beyond combinatorial auctions. As such, one could consider matroids or matroid intersections \citep{10.1145/2213977.2213991}. Here, it is known that we are required to use dynamic prices in order to achieve any constant competitive ratio. Is it possible to extend the simplified proofs in dual space also to this case?

Concerning computability, existing approaches mainly state prices directly \cite{DBLP:conf/soda/FeldmanGL15, DBLP:conf/focs/DuettingFKL17}. When applying the proof steps from Section~\ref{section:proof_competitive_ratios} constructively to our LP, one can solve for item prices. In particular, these prices are exactly the ones obtained by \cite{DBLP:conf/soda/FeldmanGL15, DBLP:conf/focs/DuettingFKL17}. Still, we leave it as an open problem if our general LP can be solved efficiently given access to e.g. demand oracles.
	\bibliography{references,dblp}
	\newpage
	\appendix
	\section{Proof of Lemma~\ref{lemma:framework_lp}}
\label{section:appendix}

\begin{proof}[Proof of Lemma~\ref{lemma:framework_lp}]
	We first split $\Ex[\mathbf{v}]{\mathbf{v}\left( \ALG(\mathbf{v}) \right)}$ into revenue and utility and bound each separately.  \\
		
	\emph{Revenue.} Let $T(\mathbf{v})$ denote the state of the set of allocated items $T$ when running the algorithm on valuation profile $\mathbf{v}$. The expected revenue of the algorithm is 
	\begin{align*} \label{equation:revenue} 
	\Ex[\mathbf{v}]{\REV(\mathbf{v},\mathbf{p})} = \Ex[\mathbf{v}]{\sum_{j \in T(\mathbf{v})} p_j  }  \enspace.
	\end{align*}
	
	\emph{Utility.} In order to lower bound the utility, consider an arbitrary buyer $i$. First, we consider an independent sample $\mathbf{v}'$ from the distribution. Note that buyer $i$ could either buy nothing and hence obtains a utility which is non-negative. Another option is to buy the bundle $\OPT_i((v_i, \mathbf{v}_{-i}')) \setminus T( ( v_i',\mathbf{v}_{-i} ) )$. Hence, the utility of buyer $i$ can be lower bounded by
	\begin{align*}
	\Ex[\mathbf{v}]{u_i(\mathbf{v}, \mathbf{p})} & \geq \Ex[\mathbf{v}, \mathbf{v}']{ \sum_{S \subseteq M} \left( v_i(S \setminus T( (v_i',\mathbf{v}_{-i}) )) - \sum_{j \in S \setminus T( (v_i',\mathbf{v}_{-i}) )} p_j \right)^+ \mathds{1}_{ S = \OPT_i( (v_i, \mathbf{v}_{-i}') )  } }  \\ & \geq \Ex[\mathbf{v}, \mathbf{v}']{ \sum_{S \subseteq M} \beta \left( v_i(S \setminus T( (v_i',\mathbf{v}_{-i}) )) - \sum_{j \in S \setminus T( (v_i',\mathbf{v}_{-i}) )} p_j \right)^+ \mathds{1}_{ S = \OPT_i( (v_i, \mathbf{v}_{-i}') )  } } \\ & \geq \Ex[\mathbf{v}, \mathbf{v}']{ \sum_{S \subseteq M} \beta \left( v_i(S \setminus T( (v_i',\mathbf{v}_{-i}) )) - \sum_{j \in S \setminus T( (v_i',\mathbf{v}_{-i}) )} p_j \right) \mathds{1}_{ S = \OPT_i( (v_i, \mathbf{v}_{-i}') )  }} \\ & = \Ex[\mathbf{v}, \mathbf{v}']{ \sum_{S \subseteq M} \beta \left( v_i(S \setminus T( \mathbf{v}') ) - \sum_{j \in S \setminus T( \mathbf{v}' )} p_j \right) \mathds{1}_{ S = \OPT_i(  \mathbf{v} )  } }  \enspace.
	\end{align*}
	
	Observe that in the second inequality we multiply a non-negative term by $\beta \in [0,1]$ and drop the $(\cdot)^+$ in the third inequality\footnote{\citet{DBLP:conf/ipco/CorreaCFPW22} use an equivalent line of arguments, but keep the $(\cdot)^+$ term in the utility instead of multiplying by $\beta$. Still, as a matter of fact, this will make the problem non-linear and hence, our arguments do not apply. Also, their arguments do not transfer to MPH-$k$ functions, but require to bound the bundle size of requested items.}. The last equality uses independence and the fact that $\mathbf{v}$ and $\mathbf{v}'$ are identically distributed. \\
	
	\emph{Combination.} As a consequence, summing over the lower bound of the utility for all $i$ and adding the revenue, we get 
	\begin{align*}
	\Ex[\mathbf{v}]{\mathbf{v}\left( \ALG(\mathbf{v}) \right)} & = \Ex[\mathbf{v}]{\REV(\mathbf{v},\mathbf{p})} + \sum_{i=1}^{n} \Ex[\mathbf{v}]{u_i(\mathbf{v}, \mathbf{p})} \\ & \geq \Ex[\mathbf{v}]{\sum_{j \in T(\mathbf{v})} p_j  }  + \sum_{i=1}^{n}  \Ex[\mathbf{v}, \mathbf{v}']{ \sum_{S \subseteq M} \beta \left( v_i(S \setminus T( \mathbf{v}') ) - \sum_{j \in S \setminus T( \mathbf{v}' )} p_j \right) \mathds{1}_{ S = \OPT_i(  \mathbf{v} )  } } \\ & \geq \min_{T \subseteq M} \left( \sum_{j \in T} p_j + \sum_{i=1}^n \Ex[\mathbf{v}]{\sum_{S \subseteq M} \beta \left( v_i(S\setminus T) - \sum_{j \in S \setminus T} p_j \right) \mathds{1}_{S = \OPT_i(\mathbf{v})} } \right) \enspace,
	\end{align*}
	where in the last inequality we lower bound the expectation by the worst possible choice for the set of allocated items $T$. 
\end{proof}
\end{document}